\documentclass[reprint,amsmath,amssymb,aps,prl]{revtex4-2}

\usepackage{graphicx}
\usepackage{dcolumn}
\usepackage{bm}
\usepackage{bbold}
\usepackage{CJK}
\usepackage{amsthm}
\theoremstyle{plain}
\newtheorem*{theorem*}{Theorem}
\newtheorem{theorem}{Theorem}

\begin{document}
\preprint{APS/123-QED}

\title{A Triangle Governs Genuine Tripartite Entanglement}

\author{Songbo Xie}
\email{sxie9@ur.rochester.edu}
\author{Joseph H. Eberly}
\affiliation{Center for Coherence and Quantum Optics, and Department of Physics and Astronomy, University of Rochester, Rochester NY 14627 USA.
}

\date{\today}

\begin{abstract}
A previously overlooked constraint for the distribution of entanglement in three-qubit systems is exploited for the first time and used to reveal a new genuine tripartite entanglement measure. It is interpreted as the area of a so-called concurrence triangle and is compared with other existing measures. The new measure is found superior to previous attempts for different reasons. A specific example is illustrated to show that two tripartite entanglement measures can be inequivalent due to the high dimensionality of the Hilbert space.
\end{abstract}

\maketitle

\noindent{\bf Introduction.}\quad 
A striking feature of modern physics is entanglement, which describes the tensorial non-biseparability of states for two or more parties that may be well-separated in location.
Following the two-party teleportation by Bennett \textit{et al.} \cite{bennett1993}, a faithful three-party teleportation protocol was invented by Karlsson and Bourennane \cite{karlsson1998} and was shown by Hillery \textit{et al.} \cite{hillery1999} to be less vulnerable to cheating and eavesdropping than the former two-party method. This established entanglement as a powerful resource in not only two-party, but also three-party or potentially even more-party systems. A multipartite entanglement (ME) measure is thus needed in order to quantify the resource. 

Entanglement measures for two-party (especially two-qubit) systems have been well studied (see \cite{hill1997, vidal1999, alonso2016}). The Schmidt decomposition for two-qubit systems allows for only one free parameter, \textit{e.g.}, the angle $\theta$ in
\begin{equation}
    |\psi\rangle=\cos\theta|00\rangle+\sin\theta|11\rangle,\quad 0\leq\theta\leq\pi/4.
\end{equation}
Therefore all bipartite measures for such systems are equivalent in the sense that they all give the same result when answering the question whether one state is more entangled than another \cite{singh2020}. But searches for measures of multipartite entanglement still encounter difficulties.

Even for three-qubit systems the situation is much more complicated. It was found by Ac\'in \textit{et al.} \cite{acin2000} that five free parameters are needed in the generalized Schmidt decomposition for a generic three-qubit system, and thus one single measure may not be sufficient in order to fully characterize the properties of multipartite entanglement (see Vidal \cite{vidal2000}).

In addition, a new significant concept, labeled as ``genuine'', has been introduced for multi-party systems. All three-qubit states were clearly separated by D\"ur, Vidal and Cirac \cite{dur2000} into four distinct classes: product states, biseparable states, the GHZ class and the W class. In the former two classes, at least one qubit is disentangled from the rest of the system. In contrast, the three qubits in GHZ class and W class are called genuinely entangled. An important background fact is that three-party teleportation may be expected to succeed if and only if the state shared by Alice, Bob and Charlie is genuinely entangled. Thus a good ME measure has to satisfy the following two conditions to be called a {\it genuine multipartite entanglement} (GME) measure. The two conditions were identified by Ma {\it{et al.}} \cite{ma2011} as:\\[0.7em]
\noindent (a) The measure must be zero for all product and biseparable states.\\
\noindent (b) The measure must be positive for all non-biseparable states (GHZ class and W class in the three-qubit case).\\[0.7em]
\noindent Only a GME measure can faithfully quantify the three-party entanglement used as a resource in the teleportation protocol.

The open difficulties make the measurement of multipartite entanglement mysterious but interesting. Here we describe further progress by advancing a better GME measure specifically for three-qubit systems.

A series of ME measures have already been invented and developed but most of them are not GME. On the one hand, examples such as multipartite monotones by Barnum and Linden \cite{barnum2001}, Schmidt measure $P$ by Eisert and Briegel \cite{eisert2001, hein2004}, global entanglement $Q$ by Meyer {\it et al.} \cite{meyer2002,brennen2003} as well as generalized multipartite concurrence $C_N$ by Carvalho {\it et al.} \cite{carvalho2004} fail to satisfy condition (a). On the other hand, the famous 3-tangle by Coffman {\it et al.} \cite{coffman2000,miyake2003}, as well as entanglement based on ``filters'' by Osterloh and Siewert \cite{osterloh2005}, GME based on PPT mixture by Jungnitsch {\it et al.} \cite{jungnitsch2011}, and the multi-party coherence advanced by Qian {\it et al.} \cite{qian2020} violate condition (b). There are also several measures based on identifying the distance between a given state and its closest product state (see examples in \cite{shimony1995, plenio2001, wei2003}). From their definitions, they violate condition (a). Here we are introducing a new measure that does satisfy both GME requirements (a) and (b). We will identify several existing GME measures and quantify the new measure's superiority to all of them.\\

\noindent{\bf Triangle Area and GME} \quad The definition of our tripartite entanglement measure uses the well-known bipartite concurrence of Wootters (see \cite{hill1997,wootters1998}). For a generic three-qubit system, when considering the entanglement between one qubit and the remaining two taken together as an ``other'' single party, we have three one-to-other bipartite entanglements, namely $C_{1(23)}$, $C_{2(31)}$ and $C_{3(12)}$, where a subscript $i$ refers to the system's $i$th qubit.

Those three bipartite entanglements were found not completely independent by Qian {\it et al.} \cite{qian2018}. In their work, the \textit{entanglement polygon inequality} states that one entanglement cannot exceed the sum of the other two,
\begin{equation}
    C_{i(jk)}\leq C_{j(ki)}+C_{k(ij)}.
\end{equation}
A stronger version for this inequality was found by Zhu and Fei in \cite{zhu2015}, where all three concurrences are replaced by their squared forms,
\begin{equation}\label{trianglerelation}
    C_{i(jk)}^2\leq C_{j(ki)}^2+C_{k(ij)}^2.
\end{equation}
An obvious geometric interpretation \cite{qian2018} for these inequalities is that the three squared (or not) one-to-other concurrences can represent the lengths of the three edges of a triangle. When referred to the squared formula \eqref{trianglerelation}, we will call it the {\it concurrence triangle}. This is shown in Fig. \ref{fig:triangle}.

There is a physical meaning for the perimeter of the concurrence triangle. It is a tripartite entanglement measure considered by Meyer and Wallach \cite{meyer2002}, and also interpreted by Brennen \cite{brennen2003}, called {\it global entanglement}. As listed in Fig. \ref{fig:table}, global entanglement is zero only for product states, and is positive for both biseparable and non-biseparable states. Thus it violates condition (a) and is not a GME measure.

The area of the concurrence triangle is another intriguing quantity. It is zero for both product and biseparate states, and thus satisfies condition (a) for GME. However, there exists one class of concurrence triangle with zero area, but corresponding to non-biseparable states. If we reckon the area as a tripartite entanglement measure, it seems to violate condition (b) and is thus not a GME. This is included in the list in Fig. \ref{fig:table}. 

\begin{figure}[t]
    \centering
    \includegraphics[width=0.3\textwidth]{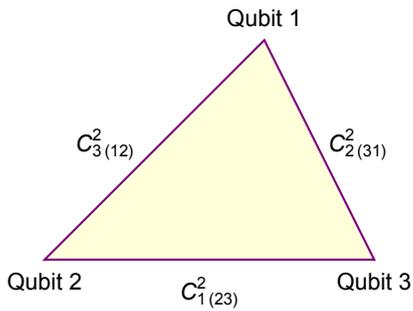}
    \caption{The concurrence triangle for a three-qubit system. The square of the three one-to-other bipartite concurrences are equal to the lengths of the three edges.}
    \label{fig:triangle}
\end{figure}

Our first result, in Theorem 1, is to show that this class of concurrence triangle does not even exist.
\begin{theorem}\label{thm:noarea}
    The area of the concurrence triangle is zero iff it has at least one edge with zero length.
\end{theorem}
\noindent This is called the {\it Triangle No-Area Theorem}. The proof is not difficult and is given as an Appendix. Generically, a triangle has zero area when its three vertices are colinear. Thm. \ref{thm:noarea} excludes the possibility that the three vertices are colinear but no two vertices coincide, which corresponds to the non-biseparable states. With this in mind, we know that the area of the concurrence triangle also satisfies condition (b). And so we have our next result:
\begin{theorem}\label{thm:fill}
    The area of the concurrence triangle is a genuine tripartite entanglement measure.
\end{theorem}
\begin{proof}
    The lengths of the three edges of the triangle are one-to-other concurrences and are thus non-increasing under local quantum operations assisted with classical communications (LQCC). Local monotonicity \cite{plenio1998} of the triangle area is naturally inherited. As is discussed above, the area satisfies both conditions (a) and (b) and is thus a GME measure.
\end{proof}

The expression for the area is given by Heron's formula
\begin{eqnarray}\label{heron}
        &F_{123}\equiv\dfrac{4}{\sqrt{3}}\sqrt{Q\left(Q-C^2_{1(23)}\right)\left(Q-C^2_{2(13)}\right)\left(Q-C^2_{3(12)}\right)},\nonumber\\
        &\text{where}\quad Q=\dfrac{1}{2}\left(C^2_{1(23)}+C^2_{2(13)}+C^2_{3(12)}\right).
\end{eqnarray}
$Q$ is the half-perimeter and thus equivalent to the global entanglement, while the prefactor $4\left/\sqrt{3}\right.$ is for normalization. We denote the area for the three-qubit triangle system as $F_{123}$, and give it a name, the \textit{concurrence fill}.

\begin{figure*}[t]
    \includegraphics[width=0.75\textwidth]{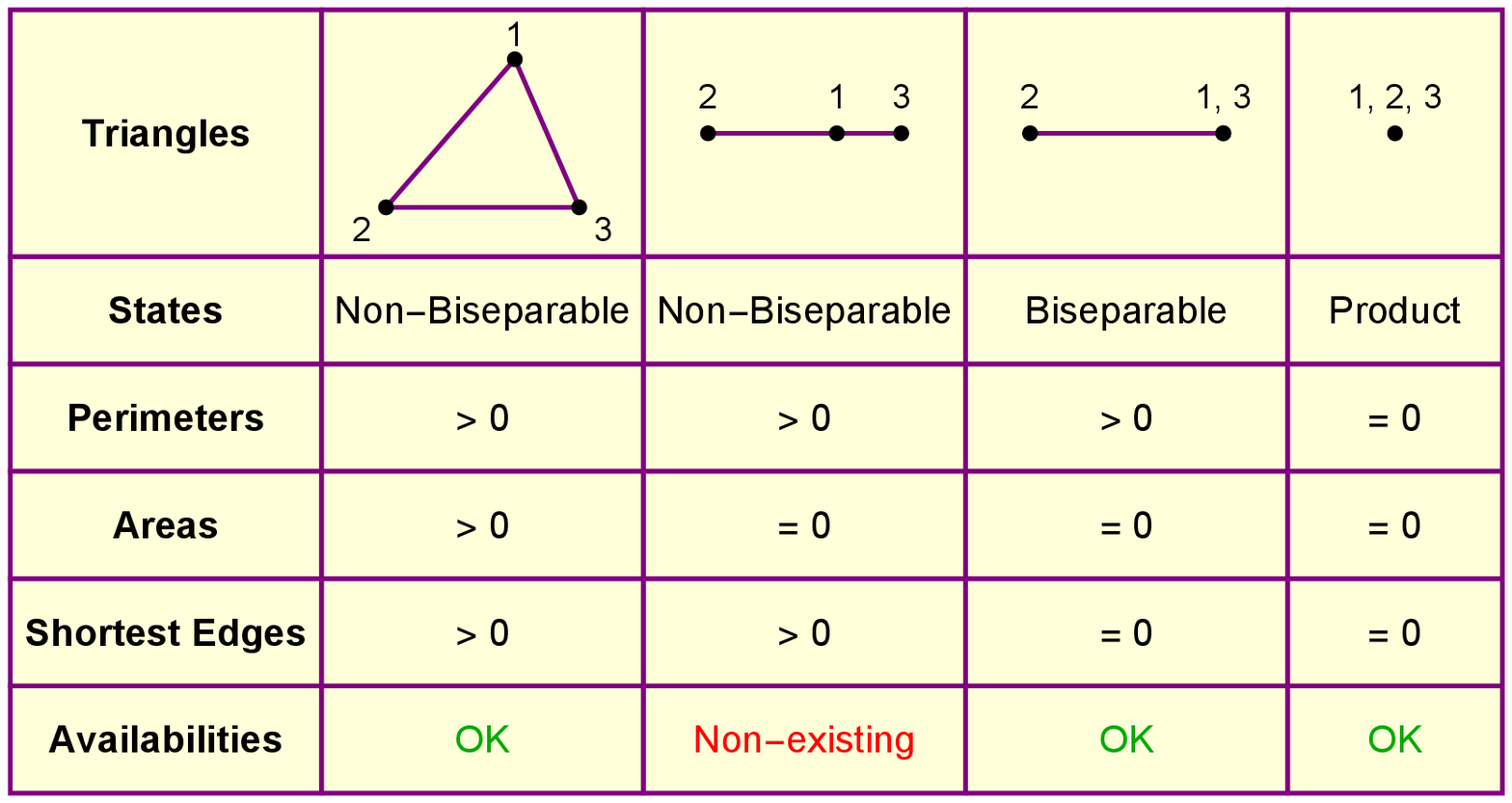}
    \caption{The table of possible states and their corresponding concurrence triangles. The values of three multipartite entanglement measures (perimeters for global entanglement, areas for concurrence fill and shortest edges for GMC) are compared with zero. One class of the non-biseparable triangles is proved to be non-existing by Thm. \ref{thm:noarea}.}
    \label{fig:table}
\end{figure*}

We provide a quick check of the $F_{123}$ measure in the following way. According to \cite{dur2000}, any pair of states in either GHZ class, or in W class, are ``stochastically equivalent'' in the sense that the conversion probability between the two states under LQCC is non-vanishing. This builds up strict rankings for the amount of entanglement within the two respective classes according to local monotonicity. However, a gap between the two classes remains since a state in GHZ class can never be converted into one in W class by LQCC, not even with only a very small probability of success, and vice versa, so there is no way to compare the entanglement for two states from the two distinct classes by using only local monotonicity. As an example, the representatives of the two classes are
\begin{eqnarray}
        |\text{GHZ}\rangle=\dfrac{1}{\sqrt{2}}\left(|000\rangle+|111\rangle\right),\nonumber\\
        |\text{W}\rangle=\dfrac{1}{\sqrt{3}}\left(|100\rangle+|010\rangle+|001\rangle\right),
\end{eqnarray}
which are the most entangled states in their respective classes. How shall we compare the entanglements for the two representatives? Helpfully for this, we employ the result shown by Joo {\it et al.} \cite{joo2003} that in three-party teleportation, the GHZ state can faithfully teleport an arbitrary single-qubit quantum state while the W state is relatively less capable, with a success rate less than 1. In this sense, we believe that one should require more than local monotonicity and conditions (a) and (b) by accepting a new condition:\\[0.7em]
\noindent (c) A GME measure ranks the GHZ state as more entangled than the W state.\\[0.7em]
\noindent Condition (c) is a bridge connecting the two distinct GHZ and W classes. A measure satisfying all the above conditions can be called a \textit{``proper'' } GME measure.

In fact, concurrence fill is maximized for the GHZ state, \textit{i.e.} $F_{123}$ =1, because the lengths of the three edges are all maximal, equal to 1. For the W state, $F_{123}$ is $64/81\approx 0.790$. The fact that concurrence fill correctly considers the GHZ state as more entangled than the W state conforms to condition (c), and thus $F_{123}$ can be regarded as a ``proper'' GME measure.\\

\noindent{\bf Comparisons of GME.}\quad Besides the ME measures mentioned in the introduction section which violate either condition (a) or (b), three GME examples already exist that satisfy both.

First {\it genuinely multipartite concurrence} (GMC), denoted $C_\text{GME}$, was advanced by Ma {\it et al.} \cite{ma2011} and further developed by Hashemi-Rafsanjani {\it et al.} \cite{rafsanjani2012}. The geometric interpretation is surprising: for three-qubit systems, $C_\text{GME}$ is exactly the square root of the length of the shortest edge of the concurrence triangle. For simplicity, in this work, we shall ignore this square root and treat $C_\text{GME}$ as the length of the shortest edge since the two resulting measures are obviously equivalent. From Fig. \ref{fig:table} we know that the shortest edge is zero for both biseparable and product states and is positive for non-biseparable states, and thus GMC is indeed a GME measure.

The second measure is the {\it generalized geometric measure} (GGM) identified by Sen(De) and Sen \cite{sende2010, sadhukhan2017}, which gives the distance between the given state and its closest biseparable state. Note that this is a generalization of the measure given by Wei and Goldbart \cite{wei2003}. GGM is quite similar to GMC in that they both give the minimal entanglement among all possible bipartitions, but with different bipartite entanglement measures. Since all bipartitions in three-qubit states must include one qubit as a subsystem, and all bipartite entanglement measures are equivalent in this one-qubit situation, GMC and GGM are equivalent for three-qubit cases. This means GMC and GGM will always give the same answer when comparing entanglements between two different three-qubit states.

The third measure is denoted as $\sigma$ by Emary and Beenakker \cite{emary2004}. Another surprising result is that $\sigma$ is actually the average of 3-tangle and GMC, {\it i.e.} $\sigma=(\tau+C_\text{GME})/2$. Thus we see that the three known measures are either equivalent or dependent. As a result, in what follows, we only need to compare concurrence fill and GMC.

In \cite{nielsen1999}, Nielsen pointed out that a pair of states in one class, although stochastically equivalent, can still be incomparable, meaning that the ranking of their entanglement cannot be judged simply by local monotonicity. We can move one step further and show that two GME measures, although both satisfying local monotonicity, can provide different opinions on the ranking of one specific pair of states, and thus are inequivalent. Indeed, concurrence fill and GMC are two inequivalent measures. In fact, for two arbitrary triangles, it is possible that one has a smaller area but a longer ``shortest edge'', while the other one has a bigger area but a shorter ``shortest edge''. Consider the following two states, both in GHZ class,
\begin{eqnarray}
       & |\psi_1\rangle = \dfrac{1}{\sqrt{2}}\sin(\dfrac{\pi}{5})|000\rangle+\dfrac{1}{\sqrt{2}}\cos(\dfrac{\pi}{5}) |100\rangle + \dfrac{1}{\sqrt{2}} |111\rangle,\nonumber \\
        & |\psi_2\rangle = \cos(\dfrac{\pi}{8})|000\rangle+\sin(\dfrac{\pi}{8})|111\rangle.
\end{eqnarray}
GMC considers that $|\psi_2\rangle$ is more entangled than $|\psi_1\rangle$ since $C_\text{GME}(\psi_2)=0.5>C_\text{GME}(\psi_1)=0.345$. However, concurrence fill considers the opposite due to the relation $F_{123}(\psi_2)=0.25<F_{123}(\psi_1)=0.393$. In this sense, GMC and concurrence fill are two inequivalent measures of tripartite entanglement. See details in Fig. \ref{fig:3tri}. Such inequivalence does not occur among two-qubit measures. It is new for three-qubit systems.

\begin{figure}[t]
    \centering
    \includegraphics[width=0.48\textwidth]{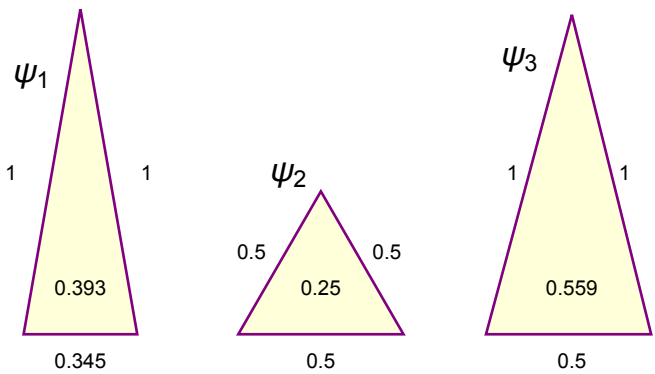}
    \caption{The concurrence triangles for $\psi_1$, $\psi_2$ and $\psi_3$ respectively. The lengths of the edges as well as the areas are shown at the appropriate locations.}
    \label{fig:3tri}
\end{figure}

By taking another glance at their definitions, one would naturally assume that concurrence fill contains more information than GMC does because $F_{123}$ depends on the lengths of all three edges but $C_\text{GME}$ only depends on the shortest one. In fact, consider the third state
\begin{equation}
    \begin{split}
        |\psi_3\rangle=&\dfrac{1}{2}|000\rangle+\dfrac{1}{2}|100\rangle+\dfrac{1}{\sqrt{2}}|111\rangle.
    \end{split}
\end{equation}
GMC cannot tell the difference between the entanglements of $|\psi_2\rangle$ and $|\psi_3\rangle$, saying that they are both 0.5, since the length of the shortest edge does not change. However, the overall triangle does change since the other two longer edges are different, but this is not detected by GMC. On the other hand, concurrence fill detects the entanglement for $|\psi_2\rangle$ as 0.25 which is much smaller than that of $|\psi_3\rangle$, given by 0.559. This can be easily visualized in Fig. \ref{fig:3tri}. In this sense, concurrence fill has an advantage over GMC.\\

\noindent{\bf Discussion and Summary.}\quad In summary, we have reviewed the conditions required by a three-party teleportation protocol that a genuine multipartite entanglement measure must satisfy. By exploiting a previously overlooked restriction for the distribution of one-to-other entanglements in the form of a theorem, we have advanced a multipartite entanglement measure $F_{123}$ for three-qubit states, with the name concurrence fill. By definition, concurrence fill is the area of the concurrence triangle, and is now shown to be a proper genuine multipartite entanglement measure, since it satisfies local monotonicity and all conditions (a), (b) and (c). It conforms to the ``proper" requirement, assigning greater entanglement to GHZ than W, which comes from the connection with the physical process of tripartite teleportation.

Finally we compared concurrence fill with another GME measure, called genuinely multipartite concurrence, which turns out to be the length of the shortest edge of the concurrence triangle. It was found that concurrence fill and GMC are two inequivalent measures in the sense that they do not always give the same result when answering the question whether one state is more entangled than the other one. A specific example was illustrated for the first time to show that two tripartite entanglement measures can be inequivalent due to the high dimensionality of the Hilbert space. One could argue that concurrence fill is a superior measure compared with GMC for the reason that it contains more information than GMC. The latter measure cannot detect the difference between entanglements of two states that are determined to carry different amounts of entanglement by concurrence fill.\\

\noindent{\bf Acknowledgements.}\quad We thank Prof. X.-F. Qian for several valuable discussions. Financial support was provided by National Science Foundation grants PHY-1501589 and PHY-1539859 (INSPIRE).

\appendix*
\section{APPENDIX: PROOF OF THEOREM \ref{thm:noarea}}
\setcounter{equation}{0}
\begin{proof}
We consider two types of bipartite entanglement, the squared concurrence $C^2$ and the normalized Schmidt weight $Y$ \cite{qian2018}. Their relation is given by
\begin{equation}
    Y(C^2)\equiv f(C^2)=1-\sqrt{1-C^2}.
\end{equation}
The first order derivative of $f(x)/x$ can be proved to be strictly positive when $x\in[0,1]$, and thus $f(x)/x$ is a strictly increasing function.

Suppose we have a concurrence triangle with zero area but none of the three edges have zero lengths. This means that the lengths of the three edges $a,b,c$ (assuming $c$ is the largest one) have to satisfy $c=a+b$, which means $c>a>0$ and $c>b>0$. This leads to
\begin{equation}\label{2inequ}
        \dfrac{f(a)}{a} < \dfrac{f(a+b)}{a+b} \quad {\rm and} \quad  \dfrac{f(b)}{b}<\dfrac{f(a+b)}{a+b}.
\end{equation}
By adding the two inequalities together, we have  $f(a)+f(b) < f(a+b)$,
or equivalently
\begin{equation}\label{smaller}
    Y(a)+Y(b) < Y(a+b) = Y(c).
\end{equation}

But remember that according to the entanglement polygon inequality in \cite{qian2018}, $Y(a)$, $Y(b)$ and $Y(c)$ are also the lengths of the three edges of a triangle, which means that
\begin{equation}\label{bigger}
    Y(a)+Y(b)\geq Y(c).
\end{equation}
Obviously, \eqref{bigger} violates \eqref{smaller}, and thus a zero area concurrence triangle cannot have all three edges with nonzero lengths. This is exactly what Thm. \ref{thm:noarea} states.
\end{proof}

\end{document}